\documentclass[11pt,reqno]{amsart}
\usepackage{txfonts}
\usepackage{mathrsfs}
\usepackage{amsfonts}
\allowdisplaybreaks[4]
\usepackage{graphicx} 
\usepackage{epsfig}
\usepackage{bbm}
\usepackage{amssymb}
\usepackage{amsmath,xypic}
\usepackage{cite,color}
\newtheorem{theorem}{Theorem}

\newtheorem{proposition}[theorem]{Proposition}

\newtheorem{lemma}[theorem]{Lemma}
\newtheorem{example}[theorem]{Example}
\newcommand{\be}{\begin{equation}}
\newcommand{\ee}{\end{equation}}
\newcommand{\bea}{\begin{eqnarray}}
\newcommand{\eea}{\end{eqnarray}}
\newcommand{\ba}{\begin{array}}
	\newcommand{\ea}{\end{array}}
\newcommand{\bean}{\begin{eqnarray*}}
	\newcommand{\eean}{\end{eqnarray*}}

\newcommand{\pa}{\partial}

\begin{document}
\title{Solutions of generalized constrained discrete KP hierarchy}
\author{Xuepu Mu$^1$, Mengyao Chen$^1$, Jipeng Cheng$^{1,2*}$, Jingsong He$^3$}
\dedicatory {$^1$ School of Mathematics, China University of Mining and Technology, Xuzhou, Jiangsu 221116, China\\
$^2$ Jiangsu Center for Applied Mathematics (CUMT), \ Xuzhou, Jiangsu 221116, China\\
$^3$ Institute for Advanced Study, Shenzhen University, \ Shenzhen, Guangdong 518060, China}
\thanks{$^*$Corresponding author. Email: chengjp@cumt.edu.cn \& chengjipeng1983@163.com.}
\begin{abstract}
Solutions of a generalized constrained discrete KP (gcdKP) hierarchy with constraint on Lax operator $L^k=(L^k)_{\geq m}+\sum_{i=1}^lq_i\Delta^{-1}\Lambda^mr_i$, are invesitigated by Darboux transformations $T_D(f)=f^{[1]}\cdot\Delta\cdot f^{-1}$ and $T_I(g)=(g^{[-1]})^{-1}\cdot\Delta^{-1}\cdot g$. Due to this special constraint on Lax operator, it is showed that the generating functions $f$ and $g$ of the corresponding Darboux transformations, can only be chosen from (adjoint) wave functions or $(L^k)_{<m}=\sum_{i=1}^lq_i\Delta^{-1}\Lambda^mr_i$. Then successive applications of Darboux transformations for gcdKP hierarchy are discussed. Finally based upon above, solutions of gcdKP hierarchy are obtained from $L^{\{0\}}=\Lambda$ by Darboux transformations.
\end{abstract}
\maketitle
\section{Introduction}
KP hierarchy \cite{Date1983,Ohta1988,van1994,Willow2004} has been playing a significant role in mathematical physics and integrable systems. As one of important generalizations for KP hierarchy, discrete KP (dKP) hierarchy \cite{Dickey1999,Adler1999,2Adler1999,Haine2000,Chen2017,Tamizhmani2000,Chen2021,Liu2024} is just the $(n-n')$--th modified KP hierarchy \cite{Jimbo1983}, which can be expressed by the bilinear equation below,
\begin{align}
{\rm Res}_{z}\tau(n,t-[z^{-1}])\tau(n',t'+[z^{-1}])e^{\xi(t-t',z)}z^{n-n'}=0,\quad n\geq n', n, n'\in \mathbb{Z} \label{dKPbilinear}
\end{align}
where ${\rm Res}_z\sum_ia_iz^i=a_{-1}$, $t=(t_1,t_2,\cdots)$, $[z^{-1}]=(z^{-1},z^{-2}/2,\cdots)$ and $\xi(t,z)=\sum_{i=1}^{+\infty}t_iz^i$. Notice that when $n'=n$, \eqref{dKPbilinear} is just the usual KP hierarchy \cite{Date1983,Ohta1988,Willow2004}, while for $n'=n+1$, we will obtain the 1st mKP hierarchy \cite{Chengjgp2018,Konopelchenko1993}. Both usual KP and 1st mKP hierarchies \cite{Date1983,Ohta1988,Konopelchenko1993} can be expressed by pseudo--differential operators $\sum_{i\ll\infty} a_i\partial^i$ with $\partial=\partial_{t_1}$. As for discrete KP, the corresponding Lax opertor $L$ is given by using shift operator $\Lambda$ defined by $\Lambda(f(n))=f(n+1)$ ({\it Frenkel type}) \cite{Dickey1999,Adler1999,2Adler1999} or difference operator $\Delta=\Lambda-1$ ({\it Khesin--Lyubashenko-- Roger (KLR) type}) \cite{Haine2000,Chen2017,Tamizhmani2000,Huang2013}, that is,
\begin{align*}
L=\left\{
    \begin{array}{ll}
      \Lambda+\sum_{j=0}^{\infty}u_{j}(n)\Lambda^{-j}, & \hbox{Frenkel type;} \\
      \Delta+\sum_{j=0}^{\infty}v_{j}(n)\Delta^{-j}, & \hbox{KLR type.}
    \end{array}
  \right.
\end{align*}
The discrete KP hierarchy discussed here is also called semi--discrete KP \cite{Li2013} or differential--difference KP \cite{Chen2017,Tamizhmani2000,Chen2021,Huang2013}, which is different from the fully discrete KP hierarchy\cite{Willow2015}.

The whole system of discrete KP hierarchy is too big to be used freely in integrable systems. In fact, reductions of discrete KP hierarchy are more applicable, since explicit equations are widely needed in the study of integrable systems. Constrained discrete KP (cdKP) hierarchy \cite{Song2023,Li2013,2Li2013,Chen2017,Chen2021,Oevel1996} is just one of famous reductions of discrete KP hierarchy defined by
\begin{align}
L^k=(L^k)_{\geq 0}+\sum_{i=1}^l q_i\Delta^{-1}r_i,\label{cdKPLax}
\end{align}
which is just the generalization of $k$--reduction $L^k=(L^k)_{\geq 0}$. The cdKP hierarchy is related with Ragnisco--Tu hierarchy \cite{Chen2017,Chen2021}, which is a kind of discretization of the AKNS hierarchy \cite{Merola1994}. For cdKP hierarchy, many important results have been investigated, such as Darboux transformations \cite{Oevel1996,Song2023,Yan2022}, bilinear equations \cite{Chen2019, Hu2020}, tau function description \cite{Song2023}, Virasoro symmetries \cite{2Li2013} and so on. But as far as we can know, there is a lack of discussion on solutions for cdKP hierarchy. In fact when discussing corresponding solutions by Darboux transformations, we usually start from one seed solution and then solve some linear differential equations. Notice that when we choose the relatively simple seed solution $L^k=\Lambda^k+\Delta^{-1}$, the corresponding linear differential equation $f_{t_p}=(L^p)_{\geq 0}(f)$ is quite difficult to solve. And also the corresponding solutions can only be for case $l=1$.

Therefore it will be very challenging to investigate the corresponding solutions of cdKP hierarchy.
Here instead, we would like to consider a generalized constrained discrete KP (gcdKP) hierarchy \cite{Oevel1996} defined by
\begin{align}
L^k=(L^k)_{\geq m}+\sum_{i=1}^lq_i\Delta^{-1}\Lambda^mr_i,\quad m\leq k,\label{gcdkpLax}
\end{align}
which is a natural extension of usual cdKP hierarchy. In fact, one can find that cdKP constraint (\ref{cdKPLax}) describes the negative part of $L^k$, while in gcdKP hierarchy, the corresponding constraint is imposed to the part of $L^k$ with $\Lambda$--order less than $m$. As for gcdKP, we believe it should be the analogue of the following constraints \cite{Oevel1998}
\begin{align}
\mathcal{L}^k=(\mathcal{L}^k)_{\geq m}+\sum_{i=1}^lq_i\partial^{-1} r_i\partial^m,\quad m=0,1,2,\label{kpmkphd}
\end{align}
where $\mathcal{L}=\sum_{i=0}^\infty f_i\partial^{1-i}$ with $\partial=\partial_{t_1}$ satisfying $\mathcal{L}_{t_j}=[(\mathcal{L}^j)_{\geq m},\mathcal{L}]$. Notice that \eqref{kpmkphd} is corresponding to constrained KP ($m=0$), constrained modified KP ($m=1$) and constrained Harry--Dym ($m=2$).
Another point we would like to point out is that the structure of gcdKP also comes from generalized bigraded Toda hierarchy \cite{Liuarxiv}, which is the analogue of the famous constrained KP hierarchy \cite{Cheng1992,Cheng1994} and can be viewed as the corresponding 2--component generalization.

In this paper, we have firstly reviewed some backgrounds on discrete KP hierarchy and its Darboux transformations \cite{Oevel1996,Liu2010} $T_D(f)=f^{[1]}\cdot\Delta\cdot f^{-1}$ and $T_I(g)=(g^{[-1]})^{-1}\cdot\Delta^{-1}\cdot g$, then we find that there are two choices for $f$ and $g$ in Darboux transformations of gcdKP hierarchy: (adjoint) wave functions (called Oevel methods \cite{Oevel1993}) and $\sum_{i=1}^lq_i\Delta^{-1}\Lambda^mr_i$ (called Aratyn methods \cite{Aratyn1995}). Then successive applications of $T_D$ and $T_I$ are discussed after explanations of two special points in Aratyn methods. After the preparation above, we find that the methods to obtain solutions for general $l$ in gcdKP hierarchy \eqref{gcdkpLax} and give some examples. Finally some conclusions and discussions are given.
\section{Discrete KP hierarchy}
In this section, we will review some basic facts on discrete KP hierarchy and some results on discrete KP Darboux transformations. One can refer to \cite{Song2023,Adler1999,2Adler1999,Liu2010} for more details.

Firstly the discrete KP (dKP) hierarchy is defined by the Lax equation below
\begin{align}
L_{t_{i}}=[(L^{i})_{\geq0},L],\quad i=1,2,\cdots \label{dKPLaxequaiton}
\end{align}
with Lax operator $L$ given by
\begin{align*}
L=\Lambda+\sum_{j=0}^{\infty}u_{j}(n)\Lambda^{-j},
\end{align*}
In this paper, we use the following symbols for $A=\sum_{i}a_{i}\Lambda^{i}$,
\begin{align*}
A_{\geq k}=\sum_{i\geq k}a_{i}\Lambda^{i}, \quad A_{<k}=\sum_{i<k}a_{i}\Lambda^{i}, \quad A_{[k]}=a_{k}, \quad A^{*}=\sum_{i}\Lambda^{-i}a_{i}.
\end{align*}
And for any function $f$, we use $f^{[l]}=f(n+l)$. The symbol $A(f)$ indicates the action of $A$ on $f$, whereas the symbol $Af$ or $A\cdot f$ will denote the operator product of $A$ and $f$.

The dKP Lax operator $L$ can be represented by dressing operator $$W=1+\sum_{j=1}^{\infty}w_{j}(n;t)\Lambda^{-j},$$ such that
\begin{align}
L=W\Lambda W^{-1}.\label{SatoformLax}
\end{align}
Then Lax equation ~\eqref{dKPLaxequaiton} is equivalent to Sato equation,
\begin{align}
W_{t_{i}}=-(L^{i})_{<0}W.\label{Satolaxi}
\end{align}

Further define wave function $w(n;t,z)$ and adjoint wave function $w^{*}(n;t,z)$ as follows,
\begin{align*}
w(n;t,z)=\Big(W(n,\Lambda)e^{\xi(t,\Lambda)}\Big)(z^{n}), \quad w^{*}(n;t,z)=\Big((W^{-1}(n,\Lambda))^{*}e^{-\xi(t,\Lambda^{-1})}\Big)(z^{-n}),
\end{align*}
satisfying
\begin{align}
L^{i}(w)&=z^{i}w,\quad \partial_{t_{i}}w=(L^{i})_{\geq 0}(w),\quad (L^{i})^{*}(w^{*})=z^{i}w^{*},\quad \partial_{t_{i}}w^{*}=-((L^{i})_{\geq 0})^{*}(w^{*}).\label{Lw}
\end{align}
Wave function $w(n;t,z)$ and adjoint wave function $w^{*}(n;t,z)$ can be generalized to eigenfunction $q$ and adjoint eigenfunction $r$, which are defined by
\begin{align*}
q_{t_{i}}=(L^{i})_{\geq 0}(q),\quad\quad r_{t_{i}}=-((L^{i})_{\geq 0})^{*}(r).\label{eigenfunction}
\end{align*}
\begin{lemma}\cite{Adler1999}\label{lemmaAB}
Let $A(n,\Lambda)=\sum_{i}a_{i}(n)\Lambda^{i},B(n,\Lambda)=\sum_{i}b_{k}(n)\Lambda^{k}$ be two operators, then
\begin{align}
A(n,\Lambda)\cdot B(n,\Lambda)^{*}=\sum_{j\in\mathbb{Z}}{\rm Res}_{z}z^{-1}\Big(A(n,\Lambda)(z^{\pm n})\cdot B(n+j,\Lambda)(z^{\mp n\mp j})\Big)\Lambda^{j},
\end{align}\label{wu}
\end{lemma}
According to ~\eqref{Lw} and Lemma \ref{lemmaAB}, we can find $w(n;t,z)$ and $w^{*}(n;t,z)$ satisfy the bilinear equation below
\begin{align}
{\rm Res}_{z}z^{-1}w(n;t,z)w^{*}(n';t',z)=0,\quad n>n'.\label{bilineareq}
\end{align}
It can be proved there exists a tau function $\tau(n;t)$ for dKP hierarchy such that
$$w(n;t,z)=\frac{\tau(n;t-[z^{-1}])}{\tau(n;t)}z^{n}e^{\xi(t,z)},\quad\quad w^{*}(n;t,z)=\frac{\tau(n+1;t+[z^{-1}])}{\tau(n+1;t)}z^{-n}e^{-\xi(t,z)},$$
Further, one can find coefficients of Lax operator $L$ can be expressed as tau functions by ~\eqref{SatoformLax} and ~\eqref{Satolaxi}, that is,
	\begin{align}
		u_{0}=\partial_{t_{1}}{\rm log}{\frac{\tau_{n+1}(t)}{\tau_{n}(t)}},\quad\quad u_{-1}=\partial^{2}_{t_{1}}{\rm log}{\tau_{n}(t)},\quad\cdots.\label{taufunction}
	\end{align}
Further, ~\eqref{bilineareq} can be written as:
\begin{align*}
{\rm Res}_{z}\tau(n,t-[z^{-1}])\tau(n',t'+[z^{-1}])e^{\xi(t-t',z)}z^{n-n'}=0,\quad n\geq n'.
\end{align*}

Given an operator $T$ and dKP Lax operator $L$, if $L^{\{1\}}=TLT^{-1}$ satisfy dKP Lax equation, that is,
$$L^{\{1\}}_{t_{1}}=[(L^{\{1\}})^{i}_{\geq 0},L^{\{1\}}],$$
then $T$ is called Darboux transformation operator of the dKP hierarchy \cite{Liu2010,Oevel1996}. Just as we know, the dKP hierarchy has two classes of basic Darboux transformation operators \cite{Liu2010,Oevel1996},
\begin{align*}
	&\bullet T_{D}(f)=f^{[1]}\cdot\Delta\cdot f^{-1},\\
	&\bullet T_{I}(g)=(g^{[-1]})^{-1}\cdot\Delta^{-1}\cdot g,
\end{align*}
where $f$ is dKP eigenfunction and $g$ is dKP adjoint eigenfunction.

\begin{proposition} \cite{Liu2010,Oevel1996,Song2023}
	Under $T_{D}(f_{1})$, the corresponding transformations of eigenfunction $f$, adjoint eigenfunction $g$, and tau function $\tau$ of the dKP hierarchy are
	\begin{align*}
		f\rightarrow  f^{\{1\}}=T_{D}(f_{1})f,\quad g\rightarrow g^{\{1\}}=(T_{D}^{-1}(f_{1}))^{*}g,\quad \tau\rightarrow \tau^{\{1\}}=f_{1}\tau.
	\end{align*}
	Under $T_{I}(g_{1})$, the corresponding transformations are
	\begin{align*}
		f\rightarrow f^{\{1\}}=T_{I}(g_{1})f,\quad g=g^{\{1\}}=(T_{I}^{-1}(g_{1}))^{*}g,\quad \tau=\tau^{\{1\}}=g_{1}^{[-1]}\tau.
	\end{align*}	
\end{proposition}
\begin{proposition} \cite{Song2023}
	The basic dKP Darboux transformation $T_{D}$ and $T_{I}$ can commute with each other, which can be shown in the figure below.
	$$
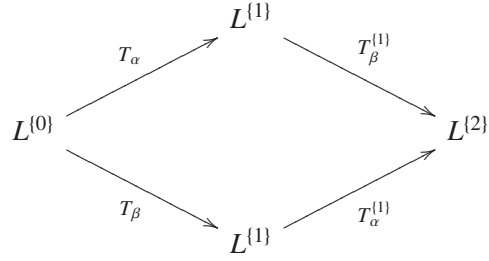
\xymatrix{&&L^{\{1\}}\ar[drr]^{T_\beta^{\{1\}}}&&\\
		L^{\{0\}}\ar[urr]^{T_\alpha}\ar[drr]_{T_\beta}& &&& L^{\{2\}}\\
		&&L^{\{1\}}\ar[urr]_{T_\alpha^{\{1\}}}&&}$$
	Here $\alpha,\beta\in\{D,I\}$, $T^{\{1\}}_{\alpha}$ and $T^{\{1\}}_{\beta}$ represents the Darboux transformation operators generated by eigenfunction or adjoint eigenfunction corresponding to $L^{\{1\}}$. The figure shows that the following relations hold,
	\begin{align}\label{TdTi}
		&T_{I}(g^{\{1\}})T_{D}(f)=T_{D}(f^{\{1\}})T_{I}(g),\nonumber\\
		&T_{D}(f^{\{1\}}_1)T_{D}(f_2)=T_{D}(f^{\{1\}}_2)T_{D}(f_1),\\
		&T_{I}(g^{\{1\}}_1)T_{I}(g_2)=T_{I}(g^{\{1\}}_2)T_{I}(g_1),\nonumber
	\end{align}
	where $A^{\{1\}}$ represents the result of $A$ under the action of the second $T$. For example,  in $T_{I}(g^{\{1\}})T_{D}(f)$, $g^{\{1\}}=(T_{D}^{-1}(f))^{*}(g)$.
\end{proposition}

Due to ~\eqref{TdTi}, we can only discuss the following Darboux chain
\begin{align*}
	&\quad L^{\{0\}} \xrightarrow{T_{D}(f_{1}^{\{0\}})} L^{\{1\}} \xrightarrow{T_{D}(f_{2}^{\{1\}})} L^{\{2\}} \xrightarrow{T_{D}(f_{3}^{\{2\}})} \cdots \xrightarrow{T_{D}(f_{M-1}^{\{M-2\}})} L^{\{M-1\}} \xrightarrow{T_{D}(f_{M}^{\{M-1\}})} L^{\{M\}}\\
	&\xrightarrow{T_{I}(g_{1}^{\{M\}})} L^{\{M+1\}} \xrightarrow{T_{I}(g_{2}^{\{M+1\}})} \cdots \xrightarrow{T_{I}(g_{N-1}^{\{M+N-2\}})} L^{\{M+N-1\}} \xrightarrow{T_{I}(g_{N}^{\{M+N-1\}})} L^{\{M+N\}},\label{Darbouxchain}
\end{align*}
where $f_{j}^{\{0\}}$ and $g_{j}^{\{0\}}$ are the corresponding eigenfunctions and adjoint eigenfunctions with respect to $L^{\{0\}}$, respectively. In what follows, we will denote
\begin{align}
	T^{\{M,N\}}(f_{1},\cdots,f_{M};g_{1},\cdots,g_{N})=T_{I}(g_{N}^{\{N+M-1\}})\cdots T_{I}(g_{1}^{\{M\}})T_{D}(f_{M}^{\{M-1\}})\cdots T_{D}(f_{1}^{\{0\}}).
\end{align}\label{yifan}
Sometimes we also use $T^{\{M,N\}}_{(\vec{M},\vec{N})}$ or $T^{\{M,N\}}$ for short, where $\vec{M}=(M,M-1,\cdots,2,1)$ and $\vec{N}=(N,N-1,\cdots,2,1)$. In Appendix, determinant formulas for $T^{\{M,N\}}$ \cite{Song2023,Liu2010} are given.
For dKP object $A$, we denote $A_{(\vec{M},\vec{N})}^{\{M,N\}}$ or $A^{\{M+N\}}$ to be the transformed object $A$ under $T^{\{M,N\}}_{(\vec{M},\vec{N})}$.

\section{Generalized constrained discrete KP hierarchy}
In this section, we firstly show the gcdKP hierarchy is well--defined, then its equivalent bilinear formulation is given. Finally, some examples of gcdKP hierarchy are given. Firstly, the $(k,m,l)$--gcdKP hierarchy \cite{Oevel1996} $(k\geq m,k\geq1,l\geq1)$ is defined by the following system:
\begin{align}
&L^{k}=\Lambda^{k}+u_{k-1}\Lambda^{k-1}+\cdots+u_{m}\Lambda^{m}+\sum_{j=1}^{l}q_{j}\Lambda^{m}\Delta^{-1}r_{j},\label{cdkp}\\
&L_{t_{i}}=[(L^{i})_{\geq0},L],\label{dkplaxeq}\\
&q_{j,t_{i}}=(L^{i})_{\geq 0}(q_{j}),\quad r_{j,t_{i}}=-((L^{i})_{\geq 0})^{*}(r_{j}),\label{eigenfunction}
\end{align}
where $\Delta=\Lambda-1$ and $\Delta^{-1}=\sum^{+\infty}_{i=1}\Lambda^{-i}$. This system is well-defined, which can be seen from following proposition.
\begin{proposition}\label{prop:gcdkp}
The $(k,m,l)$--gcdKP hierarchy ~\eqref{cdkp} satisfies
\begin{align*}
\pa_{t_i}\big(L^{k}_{<m}-\sum_{j=1}^{l}q_{j}\Lambda^{m}\Delta^{-1}r_{j})=0.
\end{align*}
\end{proposition}
    \begin{proof}
Firstly by ~\eqref{cdkp}, $L^{k}_{<m}=\sum_{j=1}^{l}q_{j}\Lambda^{m}\Delta^{-1}r_{j}.$ Then by ~\eqref{dkplaxeq}, we know
\begin{align*}
\partial _{t_{i}}(L_{<m}^{k})=(B_{k}\cdot \sum_{j=1}^{l}q\Delta^{-1}r^{[m]})_{<0}\Lambda^{m}-(\sum_{j=1}^{l}q\Delta^{-1}r^{[m]}\cdot B_{k}^{[m]})_{<0}\Lambda^{m}.
\end{align*}
Then by using ~\eqref{eigenfunction} and formulas for operator $A$ and function $f$,
\begin{align}\label{Af}
(A_{\geq0}f\Delta^{-1})_{<0}=A_{\geq0}(f)\Delta^{-1}, \quad (\Delta^{-1}fA_{\geq0})_{<0}=\Delta^{-1}A^*_{\geq0}(f),
\end{align}
we can get $\partial _{t_{i}}(L^{k}_{<m})$ is consistent with $\partial _{t_{i}}(\sum_{j=1}^{l}q\Lambda^{m}\Delta^{-1}r)$.
\end{proof}
The $(k,m,l)$--gcdKP hierarchy can also be expressed by bilinear equations.
\begin{theorem}\label{theorem:bilinear--gcdKP}
	Given $(k,m,l)$--gcdKP hierarchy, let $w(n;t,z)$ and $w^{*}(n;t,z)$ be corresponding wave function and adjoint wave function respectively, then $L^{k}_{<m}=\sum_{j=1}^{l}q_{j}\Lambda^{m}\Delta^{-1}r_{j}$ is equivalent to
	\begin{align}\label{bilinear}
		{\rm Res}_{z}z^{-1+k}w(n;t,z)w^{*}(n';t',z)=\sum_{j=1}^{l}q_{j}(n,t)r_{j}(n',t'),\quad\quad n-n'>-m.
	\end{align}
\end{theorem}
\begin{proof}
	Firstly by $\partial_{t_{i}}w(n;t,z)=L_{\geq 0}^{i}(w(n;t,z)), \partial_{t_{i}}q(n,t)=L_{\geq 0}^{i}(q(n,t))$ and Taylor expansion, we can know ~\eqref{bilinear} is equivalent to
	\begin{align*}
		{\rm Res}_{z}z^{-1+k}w(n;t,z)w^{*}(n+i,t,z)=\sum_{j=1}^{l}q_{j}(n,t)r_{j}(n+i,t),\quad\quad i<m,
	\end{align*}
	that is
	\begin{align}
		\sum_{i<m}{\rm Res}_{z}z^{-1+k}w(n;t,z)w^{*}(n+i;t,z)\Lambda^{i}=\sum_{i<m}\Big(\sum_{j=1}^{l}q_{j}(n,t)r_{j}(n+i,t)\Big)\Lambda^{i}.\label{equ}
	\end{align}
	Further by Lemma ~\ref{lemmaAB}, we can find ~\eqref{equ} is equivalent to
	\begin{align*}
		\Big(W(n,\Lambda)\Lambda^{k}W(n,\Lambda)^{-1}\Big)_{<m}=\sum_{i<m}\Big(\sum_{j=1}^{l}q_{j}(n,t)r_{j}(n+i,t)\Big)\Lambda^{i}
		=\sum_{j=1}^{l}q_{j}(n,t)\Lambda^{m}\Delta^{-1}r_{j}(n,t),
	\end{align*}
	that is $(L^{k})_{<m}=\sum_{j=1}^{l}q_{j}(n,t)\Lambda^{m}\Delta^{-1}r_{j}(n,t)$.
	So we finish this proof.
\end{proof}
Next we give some examples for the $(k,m,l)$-gcdKP hierarchy.
\begin{example}
	$\bullet$ When $k=1$, $m=0$ and $l=1$,
	\begin{align*}
		L=\Lambda+u_{0}+q\Delta^{-1}r.
	\end{align*}\label{pingguo}
	\quad\quad$1)$\quad$t_{1}$\quad flow:
	\begin{align*}	
		&\partial_{t_{1}}u_{0}=q^{[1]}r-qr^{[-1]},\\
		&\partial_{t_{1}}q=q^{[1]}+u_{0}q,\\
		&\partial_{t_{1}}r=-r^{[-1]}-u_{0}r.
	\end{align*}
	\quad\quad$2)$\quad$t_{2}$\quad flow:
	\begin{align*}
		&\partial_{t_{2}}u_{0}=q^{[2]}r+q^{[1]}ru_{0}^{[1]}+q^{[1]}ru_{0}-qr^{[-2]}-qr^{[-1]}u_{0}-qr^{[-1]}u_{0}^{[-1]},\\
		&\partial_{t_{2}}q=q^{[2]}+u_{0}^{[1]}q^{[1]}+u_{0}q^{[1]}+u_{0}^{2}q+q^{[1]}rq+qr^{[-1]}q,\\
		&\partial_{t_{2}}r=-r^{[-2]}-u_{0}^{[-1]}r^{[-1]}-u_{0}r^{[-1]}-u_{0}^{2}r-q^{[1]}r^{2}-qr^{[-1]}r.
	\end{align*}
	$\bullet$ When $k=m=1$ and $l=1$,
	$$L=\Lambda+q\Lambda\Delta^{-1}r.$$
	\quad\quad$1)$\quad$t_{1}$\quad flow:
	\begin{align*}
		&\partial_{t_{1}}q=q^{[1]}+q^{2}r,\\
		&\partial_{t_{1}}r=-r^{[-1]}-qr^{2}.
	\end{align*}
	\quad\quad$2)$\quad$t_{2}$\quad flow:
	\begin{align*}
		&\partial_{t_{1}}q=q^{[2]}-q^{2[1]}r^{[1]}-qrq^{[1]}+q^{[1]}rq+q^{2}r^{[-1]}+q^{3}r^{2},\\
		&\partial_{t_{1}}r=-r^{[-2]}+q^{[1]}r^{[1]}r^{[-1]}+qrr^{[-1]}-q^{[1]}r^{2}-qr^{[-1]}r-q^{2}r^{3}.
	\end{align*}
	$\bullet$ When $k=1$, $m=-1$ and $l=1$,
	$$L=\Lambda+u_{0}+u_{-1}\Lambda^{-1}+q\Lambda^{-1}\Delta^{-1}r.$$
	\quad\quad$1)$\quad$t_{1}$\quad flow:
	\begin{align*}
		&\partial_{t_{1}}u_{0}=u_{-1}^{[1]}-u_{-1},\\
		&\partial_{t_{1}}u_{-1}=u_{0}u_{-1}-u_{-1}u_{0}^{[-1]}+q^{[1]}r^{[-1]}-qr^{[-2]},\\
		&\partial_{t_{1}}q=q^{[1]}+u_{0}q,\\
		&\partial_{t_{1}}r=-r^{[-1]}-u_{0}r.
	\end{align*}
	\quad\quad$2)$\quad$t_{2}$\quad flow:
	\begin{align*}
		&\partial_{t_{2}}u_{0}=q^{[2]}r-qr^{[-2]}+(u_{0}^{[1]}+u_{0})u_{-1}^{[1]}-u_{-1}(u_{0}+u_{0}^{[-1]}),\\
		&\partial_{t_{2}}u_{-1}=u_{0}^{2}u_{-1}-u_{-1}u_{0}^{2[-1]}+q^{[2]}r^{[-1]}+(u_{0}^{[1]}+u_{0})q^{[1]}r^{[-1]}-qr^{[-2]}(u_{0}^{[-1]}+u_{0}^{[-2]})-qr^{[-3]},\\
		&\partial_{t_{2}}q=q^{[2]}+(u_{0}^{[1]}+u_{0})q^{[1]}+u_{0}^{2}q+u_{-1}^{[1]}q+u_{-1}q,\\
		&\partial_{t_{2}}r=-r^{[-2]}-(u_{0}^{[-1]}+u_{0})r^{[-1]}-u_{0}^{2}r-u_{-1}^{[1]}r-u_{-1}r.
	\end{align*}
\end{example}

\section{Darboux transformations for generalized constrained discrete KP hierarchy}\label{DTs}
Due to the special constraint on Lax operator, the generating functions $f$ and $g$ in dKP Darboux transformations $T_D(f)$ and $T_I(g)$ can not be chosen arbitrarily. In the section, Darboux transformations for gcdKP hierarchy are obtained by suitable choice of special $f$ and $g$, and also the corresponding successive applications are given.
\subsection{Darboux transformations for the  $(k,m,l)$--gcdKP hierarchy}
Given an operator $T$ and Lax operator $L$ of the $(k,m,l)$-gcdKP hierarchy, if $L^{\{1\}}=TLT^{-1}$ satisfies
\begin{align*}
&\partial_{t_{i}}L^{\{1\}}=[(L^{\{1\}})^{i}_{\geq 0},L^{\{1\}}],\\
&(L^{k})^{\{1\}}=(L^{k})^{{\{1\}}}_{\geq m}+\sum_{j=1}^{l}q_{j}^{\{1\}}\Lambda^{m}\Delta^{-1}r_{j}^{\{1\}},\\
&q_{j,t_{i}}^{\{1\}}=(L^{\{1\}})^{i}_{\geq 0}(q^{\{1\}}_{j}),\quad r_{j,t_{i}}=-((L^{\{1\}})^{i}_{\geq 0})^{*}(r^{\{1\}}_{j}),
\end{align*}
then $T$ is called Darboux transformation operator of the $(k,m,l)$--gcdKP hierarchy ~\eqref{cdkp}.
\begin{lemma}\label{lemmaTd}
For any operator $A$ and functions $f$,
\begin{align}
&\Big(T_{D}(f)\cdot A_{\geq m}\cdot T_{D}^{-1}(f)\Big)_{<m}=\Big(T_{D}(f)\cdot A_{\geq m}\Big)(f)\cdot\Lambda^{m}\Delta^{-1}\cdot(f^{[1]})^{-1},\label{Tdzheng}\\
&\Big(T_{D}(f)\cdot \sum_{j=1}^{l}q_{j}\Lambda^{m}\Delta^{-1}r_{j}\cdot T_{D}^{-1}(f)\Big)_{<m}=\sum_{j=1}^{l}\Big(T_{D}(f)\cdot q_{j}\Lambda^{m}\Delta^{-1}r_{j}\Big)(f)\cdot\Lambda^{m}\Delta^{-1}\cdot(f^{[1]})^{-1}\nonumber\\
&\quad\quad\quad\quad\quad\quad\quad\quad\quad\quad\quad\quad\quad\quad\quad\quad+\sum_{j=0}^{l}(T_{D}(f))(q_{j})\cdot\Lambda^{m}\Delta^{-1}\cdot(T_{D}^{-1}(f))^{*}(r_{j}).\label{Tdfu}
\end{align}
\end{lemma}
\begin{proof}
Firstly notice that
$$\Big(T_{D}(f)\cdot A_{\geq m}\cdot T_{D}^{-1}(f)\Big)_{<m}=\Big(T_{D}(f)\cdot A_{\geq m}\cdot\Lambda^{-m}f^{[m]}\Delta^{-1}\Big)_{<0}(f^{[m+1]})^{-1}\Lambda^{m},$$
then by ~\eqref{Af}, we can get ~\eqref{Tdzheng}. As for ~\eqref{Tdfu}, by using the formula for $X_{i}=f_{i}\Delta^{-1}g_{i}\ (i=1,2)$,
\begin{align*}
X_{1}\cdot X_{2}=X_{1}(f_{2})\cdot\Delta^{-1}\cdot g_{2}+f_{1}\cdot\Delta^{-1}\cdot X_{2}^{*}(g_{1}),
\end{align*}
we have
\begin{align*}
T_{D}(f)\cdot q\Lambda^{m}\Delta^{-1}r\cdot T_{D}^{-1}(f)=&\Big(T_{D}(f)\cdot q\Delta^{-1}r^{[m]}\cdot T_{D}^{-1}(f^{[m]})\Big)_{<0}\Lambda^{m}\\
=&\Big(T_{D}(f)\cdot(q\Delta^{-1}r^{[m]})(f^{[m]})\cdot\Delta^{-1}\cdot(f^{[m+1]})^{-1}\Big)_{<0}\Lambda^{m}\\
&+\Big(T_{D}(f)\cdot q\cdot\Delta^{-1}\cdot(T_{D}^{-1}(f^{[m]}))^{*}(r^{[m]})\Big)_{<0}\Lambda^{m}.
\end{align*}
Finally by ~\eqref{Af}, we can get ~\eqref{Tdfu}.
\end{proof}
Similarly, we can get the next lamma.
\begin{lemma}\label{lemmaTi}
For any operator $A$ and functions $g$,
\begin{align*}
&\Big(T_{I}(g)\cdot A_{\geq m}\cdot T_{I}^{-1}(g)\Big)_{<m}=(g^{[-1]})^{-1}\cdot\Lambda^{m}\Delta^{-1}\cdot\Big(A_{\geq m}\cdot T_{I}^{-1}(g)\Big)^{*}(g),\\
&\Big(T_{I}(g)\cdot\sum_{j=1}^{l}q_{j}\Lambda^{m}\Delta^{-1}r_{j}\cdot T_{I}^{-1}(g)\Big)_{<m}=\sum_{j=0}^{l}T_{I}(g)(q_{j})\cdot\Lambda^{m}\Delta^{-1}\cdot (T_{I}^{-1}(g))^{*}(r_{j})\\
&\quad\quad\quad\quad\quad\quad\quad\quad\quad\quad\quad\quad\quad\quad\quad+\sum_{j=1}^{l}(g^{[-1]})^{-1}\cdot\Lambda^{m}\Delta^{-1}\cdot\Big( (q_{j}\Lambda^{m}\Delta^{-1}r_{j})\cdot T_{I}^{-1}(g)\Big)^{*}(g).
\end{align*}
\end{lemma}
\begin{proposition}\label{propositionDarboux}
Under Darboux transformation operator $T=T_{D}(f)$ or $T_{I}(g)$, the Lax operator $L$ of the $(k,m,l)$--gcdKP hierarchy ~\eqref{cdkp} becomes
\begin{align}
(L^{k})^{\{1\}}=(L^{k})_{\geq m}+q^{\{1\}}_{0}\Lambda^{m}\Delta^{-1}r^{\{1\}}_{0}+\sum_{j=1}^{l}q_{j}^{\{1\}}\Lambda^{m}\Delta^{-1}r_{j}^{\{1\}}.\label{l1}
\end{align}
$\bullet$ Under the action of $T_{D}(f)$:
\begin{align}
q^{\{1\}}_{0}&=(T_{D}(f)\cdot L^{k})(f),\quad r^{\{1\}}_{0}=(f^{[1]})^{-1},\label{Tdqr0}\\
q^{\{1\}}_{j}&=(T_{D}(f))(q),\quad\quad\quad r^{\{1\}}_{j}=(T_{D}^{-1}(f))^{*}(r).\label{Tdqr1}
\end{align}
$\bullet$ Under the action of $T_{I}(g)$:
\begin{align}
q^{\{1\}}_{0}&=(g^{[-1]})^{-1},\quad\quad r^{\{1\}}_{0}=(L^{k}\cdot T_{I}^{-1}(g))^{*}(g),\label{Tiqr0}\\
q^{\{1\}}_{j}&=(T_{I}(g))(q),\quad\quad r^{\{1\}}_{j}=(T_{I}^{-1}(g))^{*}(r).\label{Tiqr1}
\end{align}
\end{proposition}
Notice that, there are $l+1$ terms for $(L^{k})^{\{1\}}_{<m}$ in ~\eqref{l1}. However, there are $l$ terms for $(L^{k})_{<m}$ in ~\eqref{cdkp}. In order to make ~\eqref{l1} consistent with ~\eqref{cdkp}, one of the terms in $(L^{k})^{\{1\}}_{<m}$ must be zero. Usually, we have two methods to do this. In the first way, we can set $$q^{\{1\}}_{0}\Lambda^{m}\Delta^{-1}r^{\{1\}}_{0}=0,$$ which is called Oevel method \cite{Oevel1993}. In this way, $$(L^{k})^{\{1\}}=(L^{k})_{\geq m}+\sum_{j=1}^{l}q_{j}^{\{1\}}\Lambda^{m}\Delta^{-1}r_{j}^{\{1\}},$$ where $q_{j}^{\{1\}}$ and $r_{j}^{\{1\}}$ is given by ~\eqref{Tdqr1} or ~\eqref{Tiqr1} .
Actually, when $T=T_{D}(f)$, we can take $f(t)=w(n;t,z)$, then $q_{0}^{\{1\}}=(T_{D}(f)\cdot L^{k})(q)=0.$ While in the case $T=T_{I}(g)$, we can set $g(t)=w^{*}(n;t,z)$, then
$r_{0}^{\{1\}}=(L^{k}\cdot T_{I}^{-1}(g))^{*}(g)=0.$

In the second way, we can set $$q^{\{1\}}_{j_{0}}\Lambda^{m}\Delta^{-1}r^{\{1\}}_{j_{0}}=0,$$ for fixed $j_{0}$, $1\leq j_{0}\leq l$, which is called Aratyn method \cite{Aratyn1995}.  In this case, $$(L^{k})^{\{1\}}=(L^{k})_{\geq m}+q^{\{1\}}_{0}\Lambda^{m}\Delta^{-1}r^{\{1\}}_{0}+\sum_{j\neq j_{0}}q_{j}^{\{1\}}\Lambda^{m}\Delta^{-1}r_{j}^{\{1\}},$$
where $q_{0}^{\{1\}}$ and $r_{0}^{\{1\}}$ are given by ~\eqref{Tdqr0} and ~\eqref{Tiqr0}, while $q_{j}^{\{1\}}$ and $r_{j}^{\{1\}}$ $(j\neq j_{0})$ are given by ~\eqref{Tdqr1} and ~\eqref{Tiqr1}.
Notice that, when $T=T_{D}(f)$, take $f(t)=q_{j_{0}}(t)$, then
$q_{j_{0}}^{\{1\}}=(T_{D}(f))(q)=0.$ While in the case $T=T_{I}(g)$, we can set $g(t)=r_{j_{0}}(t)$, then $r_{j_{0}}^{\{1\}}=(T_{I}^{-1}(g))^{*}(r)=0.$
In this method, we can let $q_{0}^{\{1\}}$ and $r_{0}^{\{1\}}$ as new $q_{j_{0}}^{\{1\}}$ and $r_{j_{0}}^{\{1\}}$ respectively.
Finally, we summarize the above results in the following theorem.
\begin{theorem}\label{theoremsigua}
	Under Darboux transformation operator $T_{D}(f)$ and $T_{I}(g)$, the Lax operator $L$ of $(k,m,l)$-gcdKP hierarchy will become
	\begin{align*}
		(L^{k})^{\{1\}}=(L^{k})^{\{1\}}_{\geq m}+\sum_{j=1}^{l}q_{j}^{\{1\}}\Lambda^{m}\Delta^{-1}r_{j}^{\{1\}},
	\end{align*}
	where $q_{j}^{\{1\}}$ and $r_{j}^{\{1\}}$ are given in the tables below.
\begin{table}[ht]
	\caption{Oevel method}
    \centering
\begin{tabular}{llllll}
	\hline\hline
	& \ \ \  $q_{j}^{\{1\}}=$ & \ \ \ \ \ $r_{j}^{\{1\}}=$&  \\
	\hline
	$f=w$ &   \ $(T_{D}(w))(q)$ & \ $(T_{D}^{-1}(w))^{*}(r)$& \\
	$g=w^{*}$  &   \  $(T_{I}(w^{*}))(q)$ & \ $(T_{I}^{-1}(w^{*}))^{*}(r)$& \\
	\hline\hline
\end{tabular}
\end{table}
\begin{table}[ht]
	\caption{Aratyn method\quad $(j\neq j_{0}, r\neq r_{0})$}
	\centering
	\begin{tabular}{llllll}
		\hline\hline
		& \  \ \ \ \ \ \ $q_{0}^{\{1\}}=$ &  \ \ \ \ \ \ \ \ \ $r_{0}^{\{1\}}=$ & \ \ \ \  \ \ \ $q_{j}^{\{1\}}=$&\ \ \ \ \ \ $r_{j}^{\{1\}}=$\\
		\hline
		$f=q_{j_{0}}$ & \  \ \ $T_{D}(q_{j_{0}})\cdot L^{k}(q_{j_{0}})$ & \ \ \ $(q_{j_{0}}^{[1]})^{-1}$ & \ \ \ \ $T_{D}(q_{j_{0}})(q)$&\ \ $(T_{D}^{-1}(q_{j_{0}}))^{*}(r)$\\
		$g=r_{j_{0}}$  &  \ \ \ $(r_{j_{0}}^{[-1]})^{-1}$ & \  \ \ $(L^{k}\cdot T_{I}^{-1}(r_{j_{0}}))^{*}(r_{j_{0}})$ & \ \ \ \  $T_{I}(r_{j_{0}})(q)$&\ \ $(T_{I}^{-1}(r_{j_{0}}))^{*}(r)$\\
		\hline\hline
	\end{tabular}
\end{table}
\end{theorem}
\subsection{The successive applications of Darboux transformation}
In the $(k,m,l)$--gcdKP hierarchy, the successive applications of Darboux transformation are different from the unconstrained case. There are two important points in Aratyn method which are showed as follows.

Firstly, the following Darboux chain
\begin{align*}
	L^{\{0\}} \xrightarrow{T_{\alpha}(f^{\{0\}})} L^{\{1\}} \xrightarrow{T_{\alpha}(f^{\{1\}})} L^{\{2\}} \xrightarrow{T_{\alpha}(f^{\{2\}})} \cdots \xrightarrow{T_{\alpha}(f^{\{M-1\}})} L^{\{M\}},
\end{align*}
is equivalent to the following chain
\begin{align*}
	L^{\{0\}} \xrightarrow{T_{\alpha}(h_{\alpha,0}^{\{0\}})} L^{\{1\}} \xrightarrow{T_{\alpha}(h_{\alpha,1}^{\{1\}})} L^{\{2\}} \xrightarrow{T_{\alpha}(h_{\alpha,2}^{\{2\}})} \cdots \xrightarrow{T_{\alpha}(h_{\alpha,M-1}^{\{M-1\}})}  L^{\{M\}},
\end{align*}
where we can take $D$ or $I$ for $\alpha$. When $T_{\alpha}=T_{D}$ (or $T_{I}$), $f$ is eigenfunction $q$ (or adjoint eigenfunction $r$) in $\sum_{j=1}^{l}q_{j}\Lambda^{m}\Delta^{-1}r_{j}$.
When $T_{\alpha}=T_{D}$, $h_{\alpha,j}=(L^{\{0\}k})^{j}(f)$, while for $T_{\alpha}=T_{I}$, $h_{\alpha,j}=(L^{\{0\}k} \Delta^{-1})^{*}(f)$. In fact, we can find $h_{\alpha,0},h_{\alpha,1},h_{\alpha,2},\cdots,h_{\alpha,M-1}$ are independent eigenfunctions (or adjoint eigenfunctions) corresponding to $L^{\{0\}}$. According to Proposition \ref{propositionDarboux}, we can obtain $f^{\{j\}}=h_{\alpha,j}^{\{j\}}.$
Therefore,
\begin{align*}
	T_{\alpha}(f^{\{M-1\}})\cdots T_{\alpha}(f^{\{1\}})T_{\alpha}(f^{\{0\}})=T^{(\alpha,M)}(h_{\alpha,0},h_{\alpha,1},\cdots,h_{\alpha,M-1}),
\end{align*}
where $(D,M)={\{M,0\}}$ and $(I,M)={\{0,M\}}$.

Secondly, since $T_{I}(r_{j_{0}}^{\{1\}})T_{D}(q_{j_{0}}^{\{0\}})=T_{D}(q_{j_{0}}^{\{1\}})T_{I}(r_{j_{0}}^{\{0\}})=1$, the following Darboux chains
\begin{align}
    L^{\{0\}} \xrightarrow{T_{D}(q_{j_{0}}^{\{0\}})} L^{\{1\}} \xrightarrow{T_{I}(r_{j_{0}}^{\{1\}})} L^{\{2\}},\quad L^{\{0\}} \xrightarrow{T_{I}(r_{j_{0}}^{\{0\}})} L^{\{1\}}\xrightarrow{T_{D}(q_{j_{0}}^{\{1\}})} L^{\{2\}}\label{multipure}
\end{align}
are trivial, which means $L^{\{0\}}=L^{\{2\}}$.
For example, we can cut off the trivial parts in the following chain $L^{\{0\}} \xrightarrow{T_{D}(q_{1}^{\{0\}})} L^{\{1\}} \xrightarrow{T_{I}(r_{1}^{\{1\}})} L^{\{2\}} \xrightarrow{T_{D}(q_{2}^{\{2\}})}  L^{\{3\}}$, and obtain the equivalent chain $L^{\{0\}} \xrightarrow{T_{D}(q_{1}^{\{0\}})} L^{\{1\}}$.

Based upon above two points, we can do any successive Darboux transformations. For example, consider the following Darboux chain
\begin{align}\label{exampleDarboux}
	&L^{\{0\}} \xrightarrow{T_{D}(w(n;t,\lambda))} L^{\{1\}} \xrightarrow{T_{I}(w^{*}(n;t,\mu)^{\{1\}})} L^{\{2\}}\xrightarrow{T_{D}(q_{1}^{\{2\}})} L^{\{3\}}\nonumber\\
    &\xrightarrow{T_{D}(q_{1}^{\{3\}})} L^{\{4\}}\xrightarrow{T_{I}(r_{1}^{\{4\}})} L^{\{5\}}\xrightarrow{T_{D}(q_{1}^{\{5\}})} L^{\{6\}}\xrightarrow{T_{I}(r_{2}^{\{6\}})} L^{\{7\}}.
\end{align}
Firstly, notice that $L^{\{3\}}=L^{\{5\}}$ by ~\eqref{multipure} and ~\eqref{TdTi}, therefore ~\eqref{exampleDarboux} can be written as the following chain
\begin{align*}
	L^{\{0\}} \xrightarrow{T_{D}(w(n;t,\lambda))} L^{\{1\}} \xrightarrow{T_{D}(h_{D,0}^{\{1\}})} L^{\{2\}}\xrightarrow{T_{D}(h_{D,1}^{\{2\}})} L^{\{3\}}
    \xrightarrow{T_{I}(w^{*}(n;t,\mu)^{\{3\}})} L^{\{4\}}\xrightarrow{T_{I}(r_{2}^{\{4\}})} L^{\{5\}},
\end{align*}
then $T^{\{3,2\}}=T_{I}(r_{2}^{\{4\}})T_{I}(w^*(n;t,\mu)^{\{3\}})T_{D}(h_{D,1}^{\{2\}})T_{D}(h_{D,0}^{\{1\}})T_{D}(w(n;t,\lambda))$. And further according to Appendix,
\begin{align*}
	&T^{\{3,2\}}=\left|\begin{array}{cccc}
		\Delta^{-1}(r_{2}w) & \Delta^{-1}(r_{2}h_{D,0}) & \Delta^{-1}(r_{2}h_{D,1}) \\
		\Delta^{-1}(w^{*}w) & \Delta^{-1}(w^{*}h_{D,0}) & \Delta^{-1}(w^{*}h_{D,1}) \\
		w & g_{D,0} & g_{D,1}\\
	\end{array}\right|^{-1} \cdot\left|\begin{array}{cccc}
		\Delta^{-1}(r_{2}w) & \Delta^{-1}(r_{2}g_{D,0}) & \Delta^{-1}(r_{2}g_{D,1}) & \Delta^{-1}r_{2} \\
		\Delta^{-1}(w^{*}w) & \Delta^{-1}(w^{*}g_{D,0}) & \Delta^{-1}(w^{*}g_{D,1}) & \Delta^{-1}w^{*} \\
		w & g_{D,0} & g_{D,1} & 1 \\
		\Delta(w) & \Delta(g_{D,0}) & \Delta(g_{D,1}) & \Delta \\
	\end{array}\right|,\\
	&((T^{\{3,2\}})^{*})^{-1}=-\Lambda^{2}\left|\begin{array}{ccccccc}
		\Delta^{-1}w & \Delta^{-1}h_{D,0} & \Delta^{-1}h_{D,1} \\
		\Delta^{-1}(r_{2w}) & \Delta^{-1}(r_{2}h_{D,0}) & \Delta^{-1}(r_{2}h_{D,1}) \\
			\Delta^{-1}(w^{*}w) & \Delta^{-1}(w^{*}h_{D,0}) & \Delta^{-1}(w^{*}h_{D,1})
	\end{array}\right|\cdot\left|\begin{array}{cccc}
	\Delta^{-1}(r_{2}w) & \Delta^{-1}(r_{2}h_{D,0}) & \Delta^{-1}(r_{2}h_{D,1}) \\
	\Delta^{-1}(w^{*}w) & \Delta^{-1}(w^{*}h_{D,0}) & \Delta^{-1}(w^{*}h_{D,1}) \\
	w & g_{D,0} & g_{D,1}\\
	\end{array}\right|^{-1}.
\end{align*}
In particular, the tau function
\begin{align*}
\tau^{\{3,2\}}=\left|\begin{array}{cccc}
	\Delta^{-1}(wr_{2}) & \Delta^{-1}(h_{D,0}r_{2}) & \Delta^{-1}(h_{D,1}r_{2}) \\
	\Delta^{-1}(ww^{*}) & \Delta^{-1}(h_{D,0}w^{*}) & \Delta^{-1}(h_{D,1}w^{*}) \\
	w & h_{D,0} & h_{D,1}\\
\end{array}\right|\cdot\tau.
\end{align*}
\section{Solutions of generalized constrained discrete KP hierarchy from $L^{\{0\}}=\Lambda$}
In this section, we will construct solutions of $(k,m,l)$--gcdKP hierarchy by Darboux transformation from $L^{\{0\}}=\Lambda$. Note that when we construct solutions for $(k,m,l)$--gcdKP hierarchy starting from $L^{\{0\}}=\Lambda$, we can find $T_{D}(w)$ and $T_{I}(w^{*})$ can not produce non-trivial solution, while $T_{D}(q)$ or $T_{I}(r)$ can only produce solutions for $l=1$. In order to overcome this problem, we need the following result.
\begin{lemma}\label{xigua}
	\begin{align}
		&(T_{D}^{-1}(f_{{2},(1,-)}^{\{1,0\}}))^{*}((f_{1}^{[1]})^{-1})=\Big((T_{D}(f_{2})(f_{1}))^{[1]}\Big)^{-1},\label{firstone}\quad (T_{I}^{-1}(g_{(1,-)}^{\{1,0\}}))^{*}((f_{1}^{[1]})^{-1})=\Big((T_{I}(g)(f_{1}))^{[1]}\Big)^{-1},\\
		&T_{D}(f_{(-,1)}^{\{0,1\}})((g_{1}^{[1]})^{-1})=\Big(((T_{D}^{-1}(f))^{*}(g_{1}))^{[1]}\Big)^{-1},\quad\quad T_{I}(g_{2,(-,1)}^{\{0,1\}})((g_{1}^{[1]})^{-1})=\Big((T_{I}^{-1}(g_{2})^{*}(g_{1}))^{[1]}\Big)^{-1},
	\end{align}
	where ``$-$" means no action of $T_{D}$ or $T_{I}$.
\end{lemma}
\begin{proof}
	Here we only prove the first equation, since others are almost similar. Firstly, note that
	$f_{{2},(1,-)}^{\{1,0\}}=T_{D}(f_{1})(f_{2})=f_{1}(n+1)\cdot\Delta(f^{-1}_{1}(n)f_{2}(n)),$ we have
	\begin{align*}
		(T_{D}^{-1}(f_{{2},(1,-)}^{\{1,0\}}))^{*}((f_{1}^{[1]})^{-1})&=-(f_{{2},(1,-)}^{\{1,0\}}(n+1))^{-1}\cdot(f_{1}^{-1}(n+1)f_{2}(n+1)),
	\end{align*}
	Further by the relation below
	$$\Delta(f^{-1}(n)f_{2}(n))=-f^{-1}_{1}(n+1)f_{2}(n+1)\cdot\Delta(f_{1}(n)f_{2}^{-1}(n))\cdot f_{1}^{-1}(n)f_{2}(n),$$
	we can finally get $(T_{D}^{-1}(f_{{2},(1,-)}^{\{1,0\}}))^{*}((f_{1}^{[1]})^{-1})=\bigg(\Big(T_{D}(f_{2})(f_{1})\Big)^{[1]}\bigg)^{-1}.$
\end{proof}
By Lemma \ref{xigua} and Proposition \ref{propositionDarboux}, we can obtain the proposition below.
\begin{theorem}\label{cdkpthe}
	Under $T_{(\vec{M},\vec{N})}^{\{M,N\}}$ (see ~\eqref{yifan}), the Lax operator $L$ of $(k,m,l)$-gcdKP hierarchy becomes
	\begin{align}
		(L^{\{M,N\}})^{k}_{<m}=&\sum_{j=1}^{M}(T_{(\vec{M},\vec{N})}^{\{M,N\}}(L^{\{0\}})^{k})(f_{j})\cdot\Lambda^{m}\Delta^{-1}\cdot \left(\Big(T_{(\vec{M}\backslash \{j\},\vec{N})}^{\{M-1,N\}}(f_{j})\Big)^{[1] }\right)^{-1}\nonumber\\
		&+\sum_{j=1}^{N}\bigg(\Big((T^{\{M,N-1\}}_{(\vec{M},\vec{N}\backslash \{j\})})^{-1*}(g_{j})\Big)^{[-1]}\bigg)^{-1}\cdot\Lambda^{m}\Delta^{-1}\cdot
		((T^{\{M,N\}}_{(M,N)})^{-1})^{*}((L^{\{0\}})^{k})^{*}(g_{j})\nonumber\\
		&+\sum_{j=1}^{M}T^{\{M,N\}}_{(\vec{M},\vec{N})}(q_{j})\cdot\Lambda^{m}\Delta^{-1}\cdot((T^{\{M,N\}}_{(\vec{M},\vec{N})})^{-1})^{*}(r_{j}).
	\end{align}
	where $\vec{M}\backslash\{j\}=\{M,\cdots,j+1,j-1,\cdots,1\},\ \vec{N}\backslash\{j\}=\{N,\cdots,j+1,j-1,\cdots,1\}$.
\end{theorem}

For dKP Lax operator $L^{\{0\}}=\Lambda$, the corresponding wave function $w_{i}$ and adjoint wave function $w_{i}^{*}$, and eigenfunction $f_{i}$ and adjoint eigenfunction $g_{i}$ usually have the forms below, that is,
\begin{align*}
	&w_{i}=\lambda_{0,i}^{n}e^{\xi(t,\lambda_{0,i})},\quad\quad\quad w_{i}^{*}=\mu_{0,i}^{-n}e^{-\xi(t,\mu_{0,i})},\\
	&f_{i}=\sum_{j=1}^{\alpha_i}a_{j,i}\lambda_{j,i}^{n}e^{\xi(t,\lambda_{j,i})},\quad g_{i}=\sum_{j=1}^{\beta_i}b_{j,i}\lambda_{j,i}^{-n}e^{-\xi(t,\lambda_{j,i})}.
\end{align*}
To get solutions for $(k,m,l)$--gcdKP hierarchy, we can firstly apply Darboux transformation operator $T^{\{M,N\}}=(f_{1},\cdots,f_{M},g_{1},\cdots,g_{N})$ to $L^{\{0\}}=\Lambda$ for $M+N=l$. In this case, $L^{\{M,N\}}$ will have the form $(L^{\{l\}})_{<m}=\sum_{i=1}^{l}q_{i}\Lambda^{m}\Delta^{-1}r_{i}$, where
\begin{align*}
	&q_{i}=(T_{(\vec{M},\vec{0})}^{\{M,0\}}L^{\{0\}})(f_{i}),\quad\quad\quad\quad\quad\quad r_{i}=\bigg(\Big(T_{(\vec{M}\backslash \{i\},\vec{0})}^{\{M-1,0\}}(f_{i})\Big)^{[1]}\bigg)^{-1}, \quad \quad\quad 1\leq i\leq M,\\
	&q_{i}=\bigg(\Big((T^{\{M,N-1\}}_{(\vec{M},\vec{N}\backslash \{i\})})^{-1*}(g_{i})\Big)^{[-1]}\bigg)^{-1},\quad\quad r_{i}=((T^{\{M,N\}}_{(\vec{M},\vec{N})})^{-1})^{*}(L^{\{0\}})^{*}(g_{i}),\quad\quad  M+1\leq i\leq l.
\end{align*}
so that we can continue applying the Darboux transformation in Theorem \ref{theoremsigua} and get various solutions.
\begin{example}
$\bullet$ When $l=1$, to obtain the corresponding solutions, we can do the following Darboux chain starting from $L^{\{0\}}=\Lambda$,
$$L^{\{0\}}\xrightarrow{T_{D}(f)}L^{\{1\}}\xrightarrow{T_{D}(f^{\{1\}})}L^{\{2\}}\xrightarrow{T_{D}(f^{\{2\}})}\cdots\xrightarrow{T_{D}(f^{\{M-1\}})}L^{\{M\}}$$
where $f$ is the dKP eigenfunction corresponding to Lax operator $L^{\{0\}}=\Lambda$, i.e., $f(n)_{t_p}=f(n+p)$,
then
\begin{align}
(L^{\{M\}})_{<m}=q\Lambda^{m}\Delta^{-1}r\label{cdkpLM}
\end{align}
with
\begin{align*}
    &q=\frac{W_{M+1}(f,f^{[k]},\cdots,f^{[k(M-1)]},f^{[kM]})}{W_{M}(f,f^{[k]},\cdots,f^{[k(M-1)]})},\\
 	&r=\frac{W_{M}(f,f^{[k]},\cdots,f^{[k(M-1)]})}{W_{M+1}(f,f^{[k]},\cdots,f^{[k(M-1)]},f^{[kM]})},\\
 	&\tau=W_{M}(f,f^{[k]},\cdots,f^{[k(M-1)]}).
 \end{align*}
and the part of $(L^{\{M\}})_{\geq m}$ can be given by ~\eqref{taufunction}.

For ($1,0,1$)--gcdKP hierarchy, the corresponding Lax operator has the following form (see examples in ~\eqref{pingguo})
$$L=\Lambda+u_0+q\Delta^{-1}r,$$
To get solutions of $q$, $r$, $\tau$ and $u_0$, let us
choose
 $$f=e^{t_1+t_2}+2^ne^{2t_1+4t_2}+3^ne^{3t_1+9t_2},$$
and take $M=2$ in ~\eqref{cdkpLM}, then we can get
\begin{align*}
 &q=\frac{4e^{6t_{1}+14t_{2}}6^{n}}{e^{5t_{1}+13t_{2}}6^{n}+e^{3t_{1}
 +5t_{2}}2^{n}+4e^{4t_{1}+10t_{2}}3^{n}},\quad r=\frac{e^{4t_{2}}3^{n+1}+e^{-t_{1}-t_{2}}2^{n+1}+e^{-2t_{1}-4t_{2}}}{e^{2t_{1}+8t_{2}}6^{n+1}+4e^{t_{1}+5t_{2}}3^{n+1}+2^{n+1}},\\
&\tau=e^{5t_{1}+13t_{2}}6^{n}+e^{3t_{1}+5t_{2}}2^{n}+4e^{4t_{1}+10t_{2}}3^{n},\\
&u_{0}=\frac{2(2e^{4_{t_{1}}+16t_{2}}72^{n}+4e^{2t_{1}+10t_{2}}18^{n}+2e^{2t_{1}+8t_{2}}24^{n}+12e^{4t_{1}+18t_{2}}54^{n}+e^{t_{1}+5t_{2}}12^{n}+3e^{5t_{1}+21t_{2}}108^{n}+12e^{3t_{1}+13t_{2}}36^{n})}{(3e^{2t_{1}+8t_{2}}6^{n}+6e^{t_{1}+5t_{2}}3^{n}+2^{n})\cdot(e^{2t_{1}+8t_{2}}6^{n}+4e^{t_{1}+5t_{2}}3^{n}+2^{n})^{2}}
 	\end{align*}
 	\end{example}
 \begin{example}
 	$\bullet$ For $(1,-1,2)$--gcdKP hierarchy, let us do  the following Darboux chain starting from $L^{\{0\}}=\Lambda$,
 	$$L^{\{0\}}\xrightarrow{T_{D}(f_{1})}L^{\{1\}}
 \xrightarrow{T_{D}(f_{2}^{\{1\}})}L^{\{2\}}\xrightarrow{T_{D}(w^{*})}L^{\{3\}},$$
 where $f_{1}=e^{t_{1}+t_{2}}+2^{n}e^{2t_{1}+4t_{2}}$,
 $f_{2}=3^{n}e^{3t_{1}+9t_{2}}+4^{n}e^{4t_{1}+16t_{2}}$,
 $w^{*}=5^{-n}e^{-5t_{1}-25t_{2}}$,
 	then the corresponding Lax operator becomes the following form $$L^{\{3\}}=\Lambda+ u_{0}+u_{1}\Lambda^{-1}+q_{1}\Lambda^{-1}\Delta^{-1}r_{1}+q_{2}\Lambda^{-1}\Delta^{-1}r_{2},$$
where
\begin{align*}
   &q_{1}=\frac{4e^{5t_{1}+19t_{2}}24^{n}+9e^{4t_{1}+16t_{2}}12^{n}}{2^{3t_{1}+3}e^{2t_{1}+10t_{2}}+9e^{t_{1}+7t_{2}}4^{n}+2e^{t_{1}+3t_{2}}6^{n}+3^{n+1}},\quad r_{1}=-\frac{(2^{n+3}e^{-7t_{1}-31t_{2}}+3e^{-8t_{1}-34t_{2}})e^{5t_{1}+25t_{2}}}{8^{n+2}e^{2t_{1}+10t_{2}}+36e^{t_{1}+7t_{2}}4^{n}+12e^{t_{1}+3t_{2}}6^{n}+3^{n+2}},\\
    &q_{2}=\frac{6e^{3t_{1}+11t_{2}}8^{n}+e^{2t_{1}+4t_{2}}6^{n}}{2^{3t_{1}+3}e^{2t_{1}+10t_{2}}+9e^{t_{1}+7t_{2}}4^{n}+2e^{t_{1}+3t_{2}}6^{n}+3^{n+1}},\quad r_{2}=\frac{6(3^{n+1}e^{-6t_{1}-26t_{2}}+8\cdot 4^{n}e^{-5t_{1}-19t_{2}})e^{5t_{1}+25t_{2}}}{8^{n+2}e^{2t_{1}+10t_{2}}+36e^{t_{1}+7t_{2}}4^{n}+12e^{t_{1}+3t_{2}}6^{n}+3^{n+2}},\\
    &\tau=\frac{5}{12}e^{-5t_{1}-25t_{2}}5^{-n}(3^{n+1}e^{4t_{1}+10t_{2}}+9\cdot 4^{n}e^{5t_{1}+17t_{2}}+2\cdot 6^{n}e^{5t_{1}+13t_{2}}+8^{n+1}e^{6t_{1}+20t_{2}}),\\
    &u_{0}=(2784e^{4t_{1}+20t_{2}}+243\cdot 48^{n}e^{2t_{1}+14t_{2}}+36\cdot 108^{n}e^{2t_{1}+6t_{2}}+64e^{2t_{1}+10t_{2}}288^{n}\\
    &\quad\quad+81e^{t_{1}+7t_{2}}288^{n}+5e^{3t_{1}+13t_{2}}144^{n+1}+9e^{5t_{1}+27t_{2}}256^{n+1}+3240e^{3t_{1}+17t_{2}}96^{n}+256e^{5t_{1}+23t_{2}}384^{n})\\
    &\quad\quad\times\frac{1}{(8^{n+2}e^{2t_{1}+10t_{2}}+9\cdot 4^{n+1}e^{t_{1}+7t_{2}}+2\cdot 6^{n+1}e^{t_{1}+3t_{2}}+3^{n+2})\cdot(2^{3n+3}e^{2t_{1}+10t_{2}}+9\cdot 4^{n}e^{t_{1}+7t_{2}}+2\cdot 6^{n}e^{t_{1}+3t_{2}}+3^{n+1})^{2}},\\
    &u_{1}=e^{5t_{1}+25t_{2}}(4e^{-3t_{1}+-15t_{2}}24^{n+1}+27e^{-4t_{1}-18t_{2}}12^{n}+6e^{-4t_{1}-22t_{2}}18^{n}+72e^{-2t_{1}-8t_{2}}32^{n}+16e^{-2t_{1}-12t_{2}}48^{n})\\
    &\quad\quad\quad\times
    \frac{1}{(2^{3n+3}e^{2t_{1}+10t_{2}}+9e^{t_{1}+7t_{2}}4^{n}+2e^{t_{1}+3t_{2}}6^{n}+3^{n+1})^{2}}.
\end{align*}
Here we have used the formula $\Delta^{-1}(z^n)=z^n/(z-1)$.
 \end{example}
\section{Conclusions and Discussions}
For $(k, m,l)$--gcdKP hierarchy, we firstly check that the corresponding definition is well--defined (see proposition \ref{prop:gcdkp}), then give an equivalent bilinear equation formulation in Theorem \ref{theorem:bilinear--gcdKP}. After that, the corresponding Darboux transformations are given in Proposition \ref{propositionDarboux} and Theorem \ref{theoremsigua}. Finally we give a systemetic method (see Theorem \ref{cdkpthe}) to construct solutions of $(k, m,l)$--gcdKP hierarchy by Darboux transformations from $L^{\{0\}}=\Lambda$. Notice that we can easily solve the linear differential equations $\partial_{t_p}f(n)=f(n+p)$ and $\partial_{t_p}g(n)=-f(n-p)$ corresponding to $L^{\{0\}}=\Lambda$, so that we can solve gcdKP hierarchy by Darboux transformations.

The method used here to get solutions of gcdKP is quite typical, which may be used in other integrable systems especially for the types of constrained KP hierarchy. We believe there should be some interesting integrable systems to be found in this kind of gcdKP hierarchy discussed here.
\section*{Appendix}
$\bullet$ $\textbf{Case}$ $M>N$,
\begin{align*}
	&T^{\{M,N\}}=\frac{1}{IW_{M,N}\left(r_{N}, \cdots, r_{1} ; q_{1}, \cdots, q_{M}\right)}\left|\begin{array}{cccc}
		\Delta^{-1}(r_{N} q_{1}) & \cdots & \Delta^{-1}(r_{N} q_{M}) & \Delta^{-1}r_{N} \\
		\vdots & \ddots & \vdots & \vdots \\
		\Delta^{-1}(r_{1} q_{1}) & \cdots & \Delta^{-1}(r_{1} q_{M}) & \Delta^{-1}r_{1} \\
		q_{1} & \cdots & q_{M} & 1 \\
		\Delta q_{1} & \cdots & \Delta q_{M} & \Delta \\
		\vdots & \ddots & \vdots & \vdots \\
		\Delta^{M-N} q_{1} & \cdots & \Delta^{M-N} q_{M} & \Delta^{M-N}
	\end{array}\right|,\\
	&((T^{\{M,N\}})^{*})^{-1}=(-1)^{M}\Lambda\left|\begin{array}{cccc}
		\Delta^{-1}q_{1}& \Delta^{-1}q_{2} & \cdots  & \Delta^{-1}q_{M} \\
		\Delta^{-1}(r_{N} q_{1}) & \Delta^{-1}(r_{N} q_{2})& \cdots  & \Delta^{-1}(r_{N}q_{M}) \\
		\vdots  & \vdots & \ddots & \vdots \\
		\Delta^{-1}(r_{1} q_{1}) & \Delta^{-1}(r_{1} q_{2})& \cdots  & \Delta^{-1}(r_{1}q_{M}) \\
		q_{1} & q_{2} &\cdots & q_{M}  \\
		\Delta q_{1} & \Delta q_{2} & \cdots & \Delta q_{M}  \\
		\vdots  & \vdots & \ddots & \vdots \\
		\Delta^{M-N-2} q_{1} & \Delta^{M-N-2} q_{2} & \cdots & \Delta^{M-N-2} q_{M}
	\end{array}\right|\cdot\frac{1}{\Lambda(IW_{M,N}(r_{N},\cdots,r_{1};q_{1},\cdots,q_{M}))},
\end{align*}
where the generalized Wronskain determinant $IW_{M,N}(r_{N},\cdots,r_{1};q_{1},\cdots,q_{M})$ is defined as follows,
\begin{equation}
	IW_{M,N}(r_{N},\cdots,r_{1};q_{1},\cdots,q_{M})
	=\left|\begin{array}{cccc}
		\Delta^{-1}(r_{N}q_{1}) & \Delta^{-1}(r_{N}q_{2})  & \cdots & \Delta^{-1}(r_{N}q_{M}) \\
		\Delta^{-1}(r_{N-1}q_{1}) & \Delta^{-1}(r_{N-1}q_{2}) & \cdots & \Delta^{-1}(r_{N-1}q_{M})\\
		\vdots & \vdots & \ddots & \vdots \\
		\Delta^{-1}(r_{1}q_{1}) & \Delta^{-1}(r_{1}q_{2}) & \cdots & \Delta^{-1}(r_{1}q_{M})\\
		q_{1} & q_{2} & \cdots & q_{M}\\
		\Delta q_{1} & \Delta q_{2} & \cdots & \Delta q_{M}\\
		\vdots & \vdots & \ddots & \vdots\\
		\Delta^{M-N-1}q_{1} & \Delta^{M-N-1}q_{2} & \cdots & \Delta^{M-N-1}q_{M}
	\end{array}\right|.\label{Wronskain}
\end{equation}
When $N=0$, ~\eqref{Wronskain} is just the Casoratian determinant
\begin{align*}
	W_{M}(q_{1},\cdots,q_{M})=\left|\begin{array}{cccc}
		q_{1} & q_{2} & \cdots & q_{M}\\
		\Delta q_{1} & \Delta q_{2} & \cdots & \Delta q_{M}\\
		\vdots & \vdots & \ddots & \vdots\\
		\Delta^{M-N-1}q_{1} & \Delta^{M-N-1}q_{2} & \cdots & \Delta^{M-N-1}q_{M}
	\end{array}\right|.
\end{align*}

$\bullet$ $\textbf{Case}$ $M=N$,
\begin{align*}
	&T^{\{M,M\}}=\frac{1}{IW_{M,M}(r_{M}, \cdots, r_{1} ; q_{1}, \cdots, q_{M})} \cdot\left|\begin{array}{cccc}
		\Delta^{-1}(r_{M} q_{1}) & \cdots & \Delta^{-1}(r_{M} q_{M}) & \Delta^{-1}r_{M} \\
		\vdots & \ddots & \vdots & \vdots \\
		\Delta^{-1}(r_{1} q_{1}) & \cdots & \Delta^{-1}(r_{1} q_{M}) & \Delta^{-1}r_{1} \\
		q_{1} & \cdots & q_{M} & 1
	\end{array}\right|,\\
	&((T^{\{M,M\}})^{*})^{-1}=\frac{1}{\Lambda(IW_{M,M}(r_{M}, \cdots, r_{1} ; q_{1}, \cdots, q_{M}))}\cdot\left|\begin{array}{cccc}
		\Delta^{-1}(q_{M}r_{1}) & \cdots & \Delta^{-1}(q_{M}r_{M}) & \Delta^{-1}q_{M}\\
		\vdots & \ddots & \vdots & \vdots \\
		\Delta^{-1}(q_{1}r_{1}) & \cdots & \Delta^{-1}(q_{1}r_{M}) &\Delta^{-1}q_{1}  \\
		r_{1} &  \cdots & r_{m} & 1
	\end{array}\right|.
\end{align*}
$\bullet$ $\textbf{Case}$ $M<N$,
\begin{align*}
	&T^{\{M,N\}}=\frac{1}{\Lambda^{-1}({IW}^{*}_{M,N}(q_{M},\cdots,q_{1};r_{1},\cdots,r_{N}))}\cdot(-1)^{N}\Lambda^{-1}\left|\begin{array}{cccccccc}
		(\Delta^{*})^{-1}r_{1} & \cdots & (\Delta^{*})^{-1}r_{N} \\
		(\Delta^{*})^{-1}(r_{2}q_{M}) & \cdots & (\Delta^{*})^{-1}(r_{N}q_{M})\\
		\vdots & \ddots & \vdots \\
	    (\Delta^{*})^{-1}(r_{1}q_{1}) & \cdots & (\Delta^{*})^{-1}(r_{N}q_{1}) \\
		r_{1} & \cdots & r_{N}\\
		\Delta^{*}r_{1} & \cdots & \Delta^{*}r_{N} \\
		\vdots & \ddots & \vdots \\
		(\Delta^{*})^{N-M-2}r_{1} & \cdots & (\Delta^{*})^{N-M-2}r_{N}\\
	\end{array}\right|,\\
	&((T^{\{M,N\}})^{*})^{-1}=\frac{1}{{IW}^{*}_{M,N}(q_{M},\cdots,q_{1};r_{1},\cdots,r_{N})}\cdot\left|\begin{array}{cccccccc}
		(\Delta^{*})^{-1}(q_{M}r_{1}) & \cdots & (\Delta^{*})^{-1}(q_{M}r_{N}) & (\Delta^{*})^{-1}q_{M} \\
		(\Delta^{*})^{-1}(q_{M-1}r_{1}) & \cdots & (\Delta^{*})^{-1}(q_{M-1}r_{N}) & (\Delta^{*})^{-1}q_{1} \\
		\vdots & \ddots & \vdots & \vdots \\
		(\Delta^{*})^{-1}(q_{1}r_{1}) & \cdots & (\Delta^{*})^{-1}(q_{1}r_{N}) & (\Delta^{*})^{-1}q_{1} \\
		r_{1} & \cdots & r_{N} & 1 \\
		\Delta^{*}(r_{1}) & \cdots & \Delta^{*}(r_{N}) & \Delta^{*}\\
		\vdots & \ddots & \vdots & \vdots\\
		(\Delta^{*})^{N-M}(r_{1}) & \cdots & (\Delta^{*})^{N-M}(r_{N}) & (\Delta^{*})^{N-M} \\
	\end{array}\right|,
\end{align*}
where the generalized inverse Wronskain determinant ${IW}^{*}_{M,N}(q_{M},\cdots,q_{1};r_{1},\cdots,r_{N})$ is below
\begin{equation*}
	{IW}^{*}_{M,N}(q_{M},\cdots,q_{1};r_{1},\cdots,r_{N})
	=\left|\begin{array}{cccc}
		(\Delta^{*})^{-1}(q_{M}r_{1}) & (\Delta^{*})^{-1}(q_{M}r_{2}) & \cdots & (\Delta^{*})^{-1}(q_{M}r_{N}) \\
		\vdots & \vdots & \ddots & \vdots \\
		(\Delta^{*})^{-1}(q_{1}r_{1}) & (\Delta^{*})^{-1}(q_{1}r_{2}) & \cdots & (\Delta^{*})^{-1}(q_{1}r_{N}) \\
		r_{1} & r_{2} & \cdots & r_{N}\\
        \Delta^{*}(r_{1}) & \Delta^{*}(r_{2})  & \cdots & \Delta^{*}(r_{N})  \\
		\vdots & \vdots & \ddots & \vdots\\
		(\Delta^{*})^{N-M-1}(r_{1}) & (\Delta^{*})^{N-M-1}(r_{2})  & \cdots & (\Delta^{*})^{N-M-1}(r_{N})
	\end{array}\right|.
\end{equation*}

\noindent{\bf Acknowledgements}:

This work is supported by National Natural Science Foundation of China (Grant Nos. 12171472, 12261072 and 12071304), Shenzhen Natural Science
Fund (the Stable Support Plan Program) (Grants 20220809163103001),  Guangdong Basic and Applied Basic Research Foundation (Grants 2024A1515013106)
and ``Qinglan Project" of Jiangsu Universities.\\

\noindent{\bf Conflict of Interest}:

The authors have no conflicts to disclose.\\

\noindent{\bf Data availability}:

Date sharing is not applicable to this article as no new data were created or analyzed in this study.

\end{document}